\documentclass[preprint,12pt]{elsarticle}

\usepackage{amssymb}
\usepackage{amsthm}

\usepackage{mathtools}
\usepackage{graphicx}

\newcommand{\be}{\begin{equation}}
\newcommand{\ee}{\end{equation}}

\newcommand{\R}{{\mathbb{R}}}
\newcommand{\N}{{\mathbb{N}}}

\newcommand{\E}{{\mathbb{E}}}

\newcommand{\e}{{\rm{e}}}
\newcommand{\mkappa}{m^\kappa}

\newcommand{\Vcm}{Z_t}

\newcommand{\Ukappa}{X^\kappa_t}
\newcommand{\Ykappa}{Y^\kappa_t}
\newcommand{\taukappa}{{\tau}^\kappa}
\newcommand{\gammakappa}{{\gamma}^\kappa}
\newcommand{\taueff}{{\tau}_{\hspace{-0.07cm}\rm{eff}}}
\newcommand{\Dkappa}{\sigma^\kappa}
\newcommand{\bz}{\mathbf{z}}

\newcommand{\p}{\mathcal{P}}

\journal{XXX}

\begin{document}

\begin{frontmatter}



\title{Centre-of-mass like superposition of Ornstein--Uhlenbeck processes:
a pathway\\ to non-autonomous stochastic differential equations\\ 
and to fractional diffusion} 

\author[lasapienza]{Mirko D'Ovidio}
\author[unibo]{Silvia Vitali}
\author[upotsdam,bcam]{Vittoria Sposini}
\author[bcam]{Oleksii~Sliusarenko}
\author[bcam,isti]{Paolo Paradisi}
\author[unibo]{Gastone Castellani}
\author[bcam,iker]{Gianni Pagnini\corref{cor}}
\cortext[cor]{Corresponding author}
\ead{gpagnini@bcamath.org}

\address[unibo]{DIFA - Department of Physics and Astronomy, 
University of Bologna, Viale B. Pichat 6/2, I-40127 Bologna, Italy}
\address[lasapienza]{Department of Basic and Applied Science for Engineering, 
Sapienza University of Rome, Via Antonio Scarpa 16, I-00161 Rome, Italy}
\address[bcam]{BCAM -- Basque Center for Applied Mathematics, 
Alameda de Mazarredo 14, E-48009 Bilbao, Basque Country, Spain}
\address[iker]{Ikerbasque -- Basque Foundation for Science, 
Calle de Mar\'ia D\'iaz de Haro 3, E-48013 Bilbao,  Basque Country, Spain}
\address[isti]{ISTI-CNR, Institute of Information Science and Technology "A. Faedo",
via Moruzzi 1, I-56124, Pisa, Italy}
\address[upotsdam]{Institute for Physics and Astronomy, University of Potsdam,
Karl-Liebknecht-Strasse 24/25, D-14476 Potsdam-Golm, Germany}

\begin{abstract}
We consider an ensemble of Ornstein--Uhlenbeck processes featuring 
a population of relaxation times and a population of noise amplitudes 
that characterize the heterogeneity of the ensemble. 
We show that the centre-of-mass like variable corresponding to this ensemble 
is statistically equivalent to a process driven by a 
non-autonomous stochastic differential equation 
with time-dependent drift and a white noise. 
In particular, the time scaling and the density function of such variable 
are driven by the population of timescales and of noise amplitudes, respectively. 
Moreover, we show that this variable is equivalent in distribution 
to a randomly-scaled Gaussian process,
i.e., a process built by the product of a Gaussian process times 
a non-negative independent random variable. 
This last result establishes a connection with the so-called 
generalized gray Brownian motion and suggests application 
to model fractional anomalous diffusion in biological systems.
\end{abstract}

\begin{keyword}

Ornstein--Uhlenbeck process \sep
Heterogeneous ensemble \sep
Superposition \sep
Center of mass \sep
Non-autonomous stochastic differential equation \sep
Randomly-scaled Gaussian process \sep
Generalized grey Brownian motion \sep
Anomalous diffusion \sep
Fractional diffusion. 


\end{keyword}

\end{frontmatter}



\section{Introduction}
\label{intro}
Let $W^\kappa_t$, $\kappa \in \{1,2, \ldots, N\}$ be a sequence of Wiener processes 
$W_t^\kappa: [0,\infty) \to \R$ with $\E_x [W_t^\kappa]=0$ and 
$\E_x [W_t^\kappa W_s^\kappa]=\min(t,s)$. 
As usual, we denote by $\E_x$ the expectation with respect to $\mathbb{P}_x$ where $x$ 
is the starting point for the process under consideration.

In this paper we consider the ensemble $\{\Ukappa\}_{\kappa}$ 
of the Ornstein--Uhlenbeck (OU) processes $\Ukappa: [0,\infty)  \to \R$ 
satisfying the stochastic differential equations (SDE)
\be
d\Ukappa = - \frac{\Ukappa}{\taukappa} dt + \sqrt{2 \, \Dkappa} \, dW_t^\kappa \,, 
\quad X^\kappa_0=x \,,
\label{OU}
\ee
where, $\forall\, \kappa$, $\taukappa > 0$ and $\Dkappa >0$. 
We call $\taukappa$ relaxation times and $\Dkappa$ noise amplitudes. 
In particular, the following equivalence in law holds true
\be
\Ukappa = e^{-\frac{t}{\taukappa}} ( x +  W_{\varphi(t)}^\kappa), 
\quad \kappa =1,2, \ldots, N \,,
\ee
where
\be 
\varphi(t) = \Dkappa \taukappa \left(e^{2t / \taukappa} -1 \right) \,.
\ee
Hence, each $\Ukappa$ is a real-valued Gaussian process 
with $\E_x (\Ukappa)=x e^{-t /\taukappa}$.  
The correlation function
\be
\E_x [X_t^\kappa X_s^\kappa] = e^{-(t+s)/\taukappa} \min(\varphi(t), \varphi(s))
\label{correlationOU}
\ee
gives the covariance function
\be 
C(t,s) = \Dkappa \taukappa \,  
\e^{-|t-s|/\taukappa} - (\Dkappa \taukappa  + x^2) \e^{-(t+s)/\taukappa} \,.
\ee
Without loss of generality, we assume that $x=0$ and therefore
\be 
C(t,s) = \Dkappa \taukappa \, \left( \e^{-|t-s|/\taukappa} - \e^{-(t+s)/\taukappa} \right) \,.
\label{covarianceOU}
\ee
For a given $\kappa$, the infinitesimal generator $A$ of $\Ukappa$ is therefore
\be
Au = - \frac{x}{\taukappa} \frac{du}{dx} + \Dkappa \frac{d^2 u}{d x^2}, 
\quad u \in C_c(\mathbb{R}, \mathbb{R}) \,,
\ee
where $C_c$ is the set of smooth and compactly supported functions on $\mathbb{R}$. 
In view of the applications in diffusion problems,
in analogy with the Langevin equation (for $\gamma, m>0$)
we introduce the mass-like variable $\mkappa > 0$ such that
\be
\frac{1}{\taukappa} = \frac{\gammakappa}{\mkappa} \,, \quad 
\sqrt{\Dkappa} = \frac{\sqrt{\sigma_0}}{\mkappa} \,,
\label{def}
\ee
where $\sigma_0> 0$ is a fixed parameter independent of $\kappa$,
and, $\forall\, \kappa$,
the pair of independent parameters $(\taukappa, \Dkappa)$ in the SDE (\ref{OU})
is now replaced by the pair $(\gammakappa, \mkappa)$.

In the following, we study the centre-of-mass like process
\be
\Vcm=\sum_\kappa \frac{\mkappa}{M} \Ukappa \,,
\label{MOD1}
\ee
with $M=\sum_\kappa \mkappa$.

We highlight that, because of the center-of-mass like formulation,
the process $\Vcm$ here considered differs
from the superposition of OU processes considered for example in
Refs. \cite{csorgo_etal-cjm-1990,lin-spa-1995}
or from the superposition of OU-type processes considered in Refs. 
\cite{leonenko_etal-s-2005,barndorffnielsen_etal-mcap-2005,grahovac_etal-jsp-2016}, 
and also the type of provided results is different.
In particular, in this paper we derive the non-autonomous SDE 
satisfied by a process equivalent in distribution to the process $\Vcm$ 
and more 
we show also that the process $\Vcm$ is equivalent in distribution to
a randomly-scaled Guassian process, i.e.,
a process built by
the product of a non-negative random variable times a Gaussian process. 
This last result can be understood in the framework of the so-called
generalized grey Brownian motion (ggBm), which can be recovered when the Gaussian
process is the fractional Brownian motion.
The ggBm is a generalization of the theory of the white noise analysis
by introducing non-Gaussian measures of Mittag--Leffler type 
\cite{mura-phd-2008,mura_etal-jpa-2008,mura_etal-itsf-2009,
grothaus_etal-jfa-2015,grothaus_etal-jfa-2016}.

Under the physical point of view, 
the considered system can be understood as a
{\it heterogeneous ensemble of Brownian particles},
i.e., a system composed 
of non-identical Brownian particles that differ in their density 
(mass devided by volume). 
The study of the centre of mass allows for estimating the average concentration and
the momentum of inertia of the ensemble by computing the mean 
and the mean square displacement, 
respectively. 
Such a system can be related to the so-called 
{\it anomalous diffusion} which refers to diffusion processes that, 
in opposition to Brownian motion,
do not show a Gaussian density function of particle displacements and 
neither a linear growth in time of the mean square displacement.
Anomalous diffusion is a widespread phenomenon 
\cite{anomaloustransport,klm} 
that requires specific statistical tools \cite{meroz_etal-phyrep-2015},
that is sometimes related to fractional diffusion \cite{metzler_etal-jpa-2004},
and that can be generated by  
the polydispersity when classical thermodynamics holds
\cite{gheorghiu_etal-pnas-2004},
i.e., with Gaussian noise as in the present case,
or by noises with long-range spatiotemporal correlations 
with even "anomalous" thermodynamics 
\cite{gheorghiu_etal-pnas-2004}, i.e., with non-Gaussian noise.
In this respect, 
the derived link between the process $\Vcm$ and the ggBm 
allows for perspective applications of the present results
to model fractional diffusion in biological systems 
in view of the promising application of the ggBm in these systems 
\cite{molina_etal-pre-2016,sposini_etal-njp-2018}.

In the next Section \ref{main}, the main theorems are stated and
their proofs are given in Section \ref{proofs}. 
In Section \ref{simulations}, 
the numerical simulations related to the main results of Section \ref{main} are shown. 
In Section \ref{conclusions} the conclusions are reported and the 
application to model anomalous diffusion in biological systems is discussed.

\section{Main results}
\label{main}
In this Section we present the two main results,
i.e., Theorem \ref{Th:nonautonomous} and Theorem \ref{Th:ggBm2}, that we obtain
for the process $\Vcm$ defined in (\ref{MOD1}).
The first concerns the determination of the non-autonomous SDE satisfied by 
a process $\Vcm^H$ that is equivalent in distribution to $\Vcm$,
and the second concerns the equivalence in distribution of the process $\Vcm$
with a randomly-scaled Gaussian process,
i.e., a process defined by the product of a non-negative random variable and
a Gaussian process. 

We introduce  the scaled process 
\be
\Ykappa=\frac{\Ukappa}{\sqrt{\Dkappa}} \,,
\label{YOU}
\ee 
for which
\be
\label{PYp}
\mathbb{P}_0(\Ykappa \in dy) = \int_{\Omega_\tau} p^\kappa(y;t|\tau) q(d\tau) dy
\,,
\ee
where the Gaussian density $p^\kappa$ is conditioned to $\taukappa$ and 
\be
\mathbb{P}(\tau^\kappa \in \Omega_\tau) = \int_{\Omega_\tau} q(d\tau) =1.
\ee
We assume that $q(d\tau)= q(\tau) d\tau$, that is $\taukappa$ has population density $q : \mathbb{R}^+ \supseteq \Omega_\tau \mapsto \mathbb{R}^+$.
We notice that $\taukappa \stackrel{d}{=} \tau$, $\forall\, \kappa$, 
where $\tau$ is distributed according with $q$. 
We denote by \lq\lq$\stackrel{d}{=}$" the equality in distribution for random variables. 

We maintain the superscript $\kappa$ in \eqref{PYp} although the sequence 
$\Ukappa$ (and therefore, the sequence \eqref{YOU}) are identically distributed. 
Indeed, we have that $\Ukappa \stackrel{d}{=} X_t$ 
and $p^\kappa(y;t|\tau) = \sqrt{\sigma^\kappa} v(\sqrt{\sigma^\kappa} y; t | \tau)$ 
where $v(x;t| \tau) dx = \mathbb{P}_0(X_t \in dx)$.

\newtheorem{theorem}{Theorem}
\begin{theorem}
\label{Th:nonautonomous}
Let $\Vcm:[0,\infty) \to \R$ be the center-of-mass like stochastic process in \eqref{MOD1}.
%
%
%
Then, $\Vcm \stackrel{d}{=} Z^*_t$ where $Z^*_t$ satisfies the following non-autonomous SDE
\be
d\Vcm^*= - \frac{\Vcm^*}{\taueff(t)}dt + \frac{\sqrt{2 \sigma_0}}{M} \, dW_t^{\rm{eff}} \,, 
\quad Z_0^*=z \,,
\label{langevintype}
\ee
with $\sqrt{\sigma_0}=\mkappa \sqrt{\Dkappa}$,
$dW_t^{\rm{eff}}=\sum_\kappa dW_t^\kappa$ such that 
$\E [(dW_t^{\rm{eff}})^2]=N \E [(dW_t^\kappa)^2]$,
and where $\tau_{\rm{eff}}(t)$ is defined as 
\be
\taueff(t)= 
\left[
\frac
{\displaystyle{\int_{-\infty}^{+\infty} \int_{\Omega_\tau}
y^2 p^\kappa(y;t|\tau) q(\tau) \, d\tau dy}}
{\displaystyle{\int_{-\infty}^{+\infty} \int_{\Omega_\tau} 
\frac{y^2}{\tau^2} p^\kappa(y;t|\tau) q(\tau) \, d\tau dy}}
\right]^{1/2} \,.
\ee
\end{theorem} 

We observe that the SDE (\ref{langevintype}) is non-autonomous because of the dependence
on $t$ of the drift coefficient through the function $\taueff(t)$,
and that the noise amplitude $\sqrt{2 \sigma_0}/M$ is not a process but
a non-negative random variable. The opposite situation with a drift
independent of $t$ and the noise amplitude given in terms of a stochastic process has been
recently studied in Ref. \cite{benth_etal-spa-2018}. 
By setting 
$$
\Lambda=\sigma_0/M^2 \,,
$$ 
from (\ref{def}), we have that $\Lambda$ is an independent non-negative random variable.

\begin{theorem}
\label{Th:ggBm1} 
Let $\{\Ukappa\}_\kappa$ be the ensemble defined in (\ref{OU}) and (\ref{def}).
Let $\Vcm:[0,\infty) \to \R$ be the center-of-mass like process defined in (\ref{MOD1}).
Let $\Ykappa$ be the scaled process in \eqref{YOU}.
Denote by $B_t^H:[0,\infty) \to \R$ the Gaussian process 
$B_t^H \stackrel{d}{=} \sqrt{\E_0 [(\Ykappa)^2]} \, W_1$, $\forall\, \kappa$. 
Then we have that , $\forall\, t>0$, as $N\to \infty$,
\be
\frac{1}{\sqrt{N}} \Vcm \stackrel{d}{\to} \sqrt{\Lambda} \, B_t^H \,.
\label{ggBmMURA0}
\ee
\end{theorem}

\begin{theorem}
\label{Th:ggBm2} 
Let $\{\Ukappa\}_\kappa$ be the ensemble defined in (\ref{OU}) and (\ref{def}).
Let $\Vcm:[0,\infty) \to \R$ be the center-of-mass like process defined in (\ref{MOD1}).
Let $\Ykappa$ be the scaled process in \eqref{YOU}.
Denote by $B_t^H:[0,\infty) \to \R$ the Gaussian process 
$B_t^H \stackrel{d}{=} \sqrt{\E_0 [(\Ykappa)^2]} \, W_1$, $\forall\, \kappa$. 
Then we have that , $\forall\, t>0$, $\forall\, N>0$, 
\be
\Vcm \stackrel{d}{=} \sqrt{N\Lambda} \, B_t^H=\Vcm^H \,.
\label{ggBmMURA}
\ee
\end{theorem}

We observe that since $\Lambda$ is a random variable that is
different for any realization of the process $B_t^H$ and 
it is independent of $t$,
the process $\Vcm^H$ is not the same process studied
in Ref. \cite{garet-spa-2000}.

\newtheorem{lemma}{Lemma}
\begin{lemma}
The stochastic process $\Vcm$ has a non-exponential correlation 
controlled by the population of $\taukappa$ distributed according to the density $q(\tau)$. 
The exponential correlation is recovered in the case $q(\tau)=\delta(\tau-\tau_0)$.
\end{lemma}
\begin{proof}
Consider the scaled process
\be
\Ykappa = \frac{\Ukappa}{\sqrt{\Dkappa}} \,,
\ee
which satisfies the SDE
\be
d\Ykappa = - \frac{\Ykappa}{\taukappa} dt + \sqrt{2} \, dW^\kappa_t \,,
\label{eqOU0}
\ee
then it holds
\be
\Vcm=
\frac{\sqrt{\sigma_0}}{M} \sum_\kappa \Ykappa \,.
\label{ggBm}
\ee
Each stochastic process $\Ykappa$ is an OU process 
with exponential correlation (\ref{correlationOU}), i.e.,
$\displaystyle{\E [\Ykappa Y_s^\kappa | \tau^\kappa] = \tau^\kappa \, \left( \e^{-|t-s|/ \tau^\kappa} - \e^{-(t+s)/ \tau^\kappa} \right)}$, $\forall\, \kappa$, according to (\ref{eqOU0}).
The correlation of the process $\Vcm$ is 
\begin{eqnarray}
\E [\Vcm Z_s]
&=& \E \left[\frac{\sigma_0}{M^2} \, \E \left[\sum_\kappa \Ykappa \sum_\kappa Y_s^\kappa \bigg| \taukappa \right] \right] \nonumber \\ 
&=& \E \left[\frac{\sigma_0}{M^2}\, \sum_\kappa \E \left[\Ykappa Y_s^\kappa \big| \taukappa \right] \right] \nonumber \\ 
&=& \E \left[\frac{\sigma_0}{M^2} \sum_\kappa \taukappa \, \left( \e^{-|t-s|/ \tau^\kappa} - \e^{-(t+s)/ \tau^\kappa} \right) \right] \nonumber \\
&=& \E \left[\frac{\sigma_0}{M^2} \right]  \sum_\kappa \E \left[ \taukappa \, \left( \e^{-|t-s|/ \tau^\kappa} - \e^{-(t+s)/ \tau^\kappa} \right) \right] \nonumber \\
&=& N \, \E \left[\frac{\sigma_0}{M^2}\right] 
\int_{\Omega_\tau} \tau \, \left( \e^{-|t-s|/ \tau} - \e^{-(t+s)/ \tau} \right) \, q(d \tau),
\label{correlazione}
\end{eqnarray}
where, we recall that $\Lambda=\sigma_0/M^2$ is an independent random variable and $\taukappa
$ are independent and identically distributed. The distribution $q(\tau)$ modifies the classical correlation displayed, $\forall\, \kappa$, by $\E [\Ykappa Y_s^\kappa \big| \tau]$. When $q(\tau)=\delta(\tau-\tau_0)$, the exponential correlation follows.
\end{proof}

From Theorem \ref{Th:ggBm2}, we have that
\be
\E [\Vcm Z_s]  = N \E [\Lambda] \, \E [B_t^H B_s^H]= N \E [\Lambda] \, R(t,s) \,,
\ee
and then, after re-scaling $N\Lambda \to \Lambda$, the following holds

\begin{lemma}
Let $f(\lambda)$, $\lambda \in \Omega_\lambda \subseteq \R^+$,
be the density of the non-negative random variable $\Lambda$,
then the $n$-dimensional density of $\Vcm$ is
\be
\p(\bz;R)=\frac{1}{\sqrt{(2\pi\lambda)^n \, {\rm det} \, R}} 
\int_{\Omega_\lambda} \exp\left\{-\frac{1}{2\lambda} \, \bz^T R^{-1} \bz\right\} 
\, f(\lambda) \, d\lambda \,,
\label{ggBmpdfn}
\ee
where $\bz=(z_1,\dots,z_n)$ and $R=R(t_i,t_j)$, $i,j=1,2,\dots,n$, 
is the covariance matrix of $B_t^H$.
\end{lemma}

\newtheorem{corollary}{Corollary}
\begin{corollary}
Let $R(t,t)=t^{2H}$,
with $0 < H < 1$, and $\Lambda \sim f(\lambda)$,
then the one-point one-time density function of $\Vcm$ is  
\be
\p(z;t) 
= \frac{1}{\sqrt{4 \pi \lambda \, t^{2H}}} \int_{\Omega_\lambda} 
\exp\left\{-\frac{z^2}{4 \lambda \, t^{2H}}\right\}
\, f(\lambda) \, d\lambda \,.
\ee
\end{corollary}

\newtheorem{remark}{Remark}
\begin{remark}
If $f(\lambda)$ is the density function of $\Lambda$, then from the relation
$\Lambda=\sigma_0/M^2$ we have that the density function of $M$ is 
$$
g(M)= \frac{2 \sigma_0}{M^2} \, f\!\left(\frac{\sigma_0}{M^2}\right) \,.
$$
\end{remark}

\begin{remark}
If $\mkappa$, $\kappa \in \N$, are independent identically distributed variables
according to the distribution $\rho(m)$,
then $M=\sum \mkappa$ follows the same distribution, i.e., 
$$
g(M) \equiv \rho(m) \,.
$$ 
\end{remark}

\begin{corollary}
[See Theorem 2.1 and Corollary 2.1 in Ref. \cite{pagnini_etal-fcaa-2016}]
Let $R(t,t)=t^{2\beta/\alpha}$ and 
$f(\lambda)=K^{-\alpha/2}_{\alpha/2,\beta}(\lambda)$, with 
\cite{mainardi_etal-fcaa-2001}
\be
K^{-\theta}_{\alpha,\beta}(\lambda)=
\frac{1}{2 \pi} \int_{-\infty}^{+\infty} \e^{- i \omega \lambda} 
\left\{
\frac{1}{2 \pi i} \int_{c-i\infty}^{c+i\infty} 
\e^{pt} \frac{p^{\beta-1}}{p^\beta + \Psi_\alpha^\theta(\omega)} \, dp 
\right\} \, d\omega \,,
\ee
where $0 < \beta < 1$, $0 < \alpha < 2$, $\theta=\min\{\alpha,2-\alpha\}$ and 
$\Psi_\alpha^\theta(\omega)=|\omega|^\alpha \e^{i ({\rm{sgn}} \omega)\theta \pi/2}$,
then the one-point one-time density function of $\Vcm$ is  
\cite{pagnini_etal-fcaa-2016} 
\begin{eqnarray}
\p(z;t) 
&=& \frac{1}{\sqrt{4 \pi \lambda \, t^{2\beta/\alpha}}}\int_0^{+\infty} 
\exp\left\{-\frac{z^2}{4 \lambda \, t^{2\beta/\alpha}}\right\} 
\, K^{-\alpha/2}_{\alpha/2,\beta}(\lambda) \, d\lambda 
\nonumber \\
&=& \frac{1}{t^{\beta/\alpha}} \, K^0_{\alpha,\beta}\left(\frac{|z|}{t^{\beta/\alpha}}\right) \,.
\label{stpdf}
\end{eqnarray}
\end{corollary}

The density function $\p(z;t)$ defined in (\ref{stpdf}) is the 
kernel of the following space-time fractional diffusion equation
\cite{mainardi_etal-fcaa-2001,pagnini_etal-fcaa-2016}
\be
\left\{
\begin{array}{lr}
_tD_*^\beta \, u = {_zD^\alpha_0} \, u \,, & {\rm{in}} \,\, \Omega \,, \\
\\
u(z,0)=u_0(z) \,, & {\rm{in}} \,\, \R \,,
\end{array}
\right.
\label{stfde}
\ee
with $(z,t) \in \Omega = \R \times (0,\infty)$,
$0 < \alpha < 2$, $0 < \beta < 1$,
and where $_tD_*^β$ is the Caputo time-fractional derivative of order $\beta$
defined through the Laplace transfor pair 
$$
\int_0^{+\infty} \e^{-st} \{{_tD_*^\beta} u(z,t)\} \, dt=
s^\beta \, \widetilde{u}(z,s) -\sum_{j=0}^{m-1} s^{\beta-1-j} \, u^{(j)}(z,0^+) \,,
$$
with $m-1< \beta < m$, $m \in \N$, 
and $_zD_\theta^\alpha$ is the Riesz--Feller space-fractional derivative 
of order $\alpha$, and asymmetry parameter $\theta=\min\{\alpha,2-\alpha\}$, 
defined through the Fourier transform pair
$$
\int_{-\infty}^{+\infty}
\e^{+ i \xi z}
\{{_zD_\theta^\alpha} u(z,t)\} \, dz
= -|\xi|^\alpha \, \e^{i (\rm{sign} \, \xi)\theta\pi/2} \, \widehat{u}(\xi,t) \,.
$$

\begin{corollary}
The density function $\p(z;t)$ in (\ref{stpdf})
solves the space-time fractional diffusion equation (\ref{stfde}),
whose special cases are  
the time-fractional diffusion ($\alpha=2$), 
the space-fractional diffusion ($\beta=1$),
the neutral fractional diffusion ($\alpha=\beta$)
and the classical diffusion ($\alpha=2$, $\beta = 1$). 
\end{corollary}

The stochastic process $\Vcm^H$ defined in (\ref{ggBmMURA}) 
is built with the same constructive approach adopted 
by Mura \cite{mura-phd-2008} 
to build up the ggBm 
\cite{mura-phd-2008,mura_etal-jpa-2008,mura_etal-itsf-2009},
i.e., a Gaussian process times a non-negative random variable,
and here we call such type processes:
randomly-scaled Gaussian processes. 
The ggBm proposed by Mura is recovered from (\ref{ggBmMURA}) 
in the case $B_t^H$ is the fractional Brownian motion 
and $\Lambda \sim M_\beta(\lambda)$ where $M_\beta(\lambda)$, $0 < \beta < 1$, 
is the M-Wright/Mainardi function \cite{mainardi_etal-ijde-2010,pagnini-fcaa-2013},
i.e., $M_\beta(\lambda)=K^{-1}_{1,\beta}(\lambda)$.
Then we have the following

\begin{corollary}
Let $R(t,t)=t^{2H}$
and $f(\lambda)= M_\beta(\lambda)$,
with $0 < H < 1$ and $0 < \beta < 1$, 
then the one-point one-time density function of $\Vcm$ is  
\cite{pagnini-fcaa-2012} 
\begin{eqnarray}
\p(z;t) 
&=& \frac{1}{\sqrt{4 \pi \lambda \, t^{2H}}}\int_0^{+\infty} 
\exp\left\{-\frac{z^2}{4 \lambda \, t^{2H}}\right\}
\, M_\beta(\lambda) \, d\lambda 
\nonumber \\
&=& \frac{1}{2 \, t^H} \, M_{\beta/2}\left(\frac{|z|}{t^H}\right) \,.
\label{fbmpdf}
\end{eqnarray}
\end{corollary}

The density function $\p(z;t)$ defined in (\ref{fbmpdf}) is the 
kernel of the following Erd\'elyi--Kober fractional diffusion equation
\cite{pagnini-fcaa-2012}
\be
\left\{
\begin{array}{lr}
\displaystyle
\frac{\partial u}{\partial t} = 
\frac{2H}{\beta} t^{2H-1} \, \mathcal{D}_{2H/\beta}^{\beta-1,1-\beta} 
\frac{\partial^2 u}{\partial x^2} \,, & {\rm{in}} \,\, \Omega \,, \\
\\
u(z,0)=u_0(z) \,, & {\rm{in}} \,\, \R \,,
\end{array}
\right.
\label{ekfde}
\ee
with $(z,t) \in \Omega = \R \times (0,\infty)$,
$0 < H < 1$, $0 < \beta < 1$,
where $\mathcal{D}_{\eta}^{\sigma,\mu}$ is 
the Erd\'elyi--Kober fractional derivative defined as follows.
Let $n-1 < \mu \le n$, $n \in \N$, $\eta > 0$, $\gamma \in \R$,
the Erd\'elyi--Kober fractional derivative is defined as 
$$
D^{\gamma,\mu}_\eta \, \varphi(t)=\prod_{j=1}^n 
\left(\gamma+j+\frac{1}{\eta}t\frac{d}{dt}\right)(I^{\gamma+\mu,n-\mu}_\eta \, \varphi(t)) \,,
$$
where
$$
I^{\gamma,\mu}_{\eta} \, \varphi(t)
=\frac{\eta}{\Gamma(\mu)} t^{-\eta(\mu+\gamma)}
\int_0^t \tau^{\eta(\gamma+1)-1} (t^\eta-\tau^\eta)^{\mu-1} \varphi(\tau) \, d\tau \,.
$$

\begin{corollary}
The density function $\p(z;t)$ in (\ref{fbmpdf})
solves a fractional diffusion equation
in the Erd\'elyi--Kober sense (\ref{ekfde}),
whose special cases are  
the time-fractional diffusion ($\beta = 2H$), 
the Brownian non-Gaussian motion ($2H=1$),
the Gaussian non-Brownian motion ($\beta = 1$)
and the classical diffusion ($\beta = 2H = 1$). 
\end{corollary}

\section{Proofs of Theorems \ref{Th:nonautonomous}, \ref{Th:ggBm1} and \ref{Th:ggBm2}}
\label{proofs}
%

\subsection{Proof of Theorem \ref{Th:nonautonomous}}
The stochastic dynamics of the process $\Vcm$ is governed by  
$$
d\Vcm = \sum_\kappa \frac{\mkappa}{M} d\Ukappa \,.
\label{MOD10}
$$
By remembering the definitions in (\ref{def}),
if we plug (\ref{OU}) into (\ref{MOD1}), 
we can derive the following Langevin-type dynamics:
\begin{eqnarray*}
d\Vcm 
&=& \frac{1}{M} \sum_\kappa \left\{ - \gammakappa \Ukappa dt 
+ \sqrt{2 \sigma_0} \, dW_t^\kappa\right\} \\
&=& - \frac{1}{M} \left\{\sum_\kappa \gammakappa \Ukappa\right\} dt 
+ \frac{\sqrt{2 \sigma_0}}{M} \sum_\kappa dW_t^\kappa \,.
\end{eqnarray*}
We introduce now the scaled process $\Ykappa = \Ukappa/\sqrt{\Dkappa}$,
then we have
\begin{eqnarray}
\sum_\kappa \gammakappa \, \Ukappa
&=& \sum_\kappa \gammakappa \sqrt{\Dkappa} \, \Ykappa  
= \sqrt{\sigma_0} \sum_\kappa \frac{\Ykappa}{\taukappa} \nonumber \\
& \stackrel{d}{=} & \sqrt{\sigma_0} \, \sqrt{\sum_\kappa 
\E \left[\left(\frac{\Ykappa}{\taukappa}\right)^2\right]} \, W_1 \nonumber \\
&=& \sqrt{\sigma_0} \, 
\frac{\sqrt{\sum_\kappa \E \left[\left(\frac{\Ykappa}{\taukappa}\right)^2\right]}}
{\sqrt{\sum_\kappa \E [(\Ykappa)^2]}} 
\sqrt{\sum_\kappa \E [(\Ykappa)^2]} \, W_1 \nonumber \\
&=& \frac{1}{\taueff} \, \sqrt{\sum_\kappa \E [(\sqrt{\sigma_0} \Ykappa)^2]} \, W_1 \nonumber \\
&=& \frac{1}{\taueff} \, \sqrt{\sum_\kappa \E [(\mkappa \Ukappa)^2]} \, W_1 \nonumber \\
& \stackrel{d}{=} & \frac{1}{\taueff} \sum_\kappa \mkappa \Ukappa = M \, \frac{\Vcm}{\taueff} \,,
\label{drift}
\end{eqnarray}
where  
\be
\frac{1}{\taueff}
=
\frac{\sqrt{\sum_\kappa \E \left[\left(\frac{\Ykappa}{\taukappa}\right)^2\right]}}
{\sqrt{\sum_\kappa \E [(\Ykappa)^2]}} 
=
\left[
\frac{\displaystyle{\int_{-\infty}^{+\infty} \int_{\Omega_\tau} 
\frac{y^2}{\tau^2} p^\kappa(y;t|\tau) q(\tau) \, d\tau dy}}
{\displaystyle{\int_{-\infty}^{+\infty} \int_{\Omega_\tau}
y^2 p^\kappa(y;t|\tau) q(\tau) \, d\tau dy}}
\right]^{1/2} 
\,,
\label{gammaeff}
\ee
which is a function of time, i.e., $\tau_{\rm{eff}}=\tau_{\rm{eff}}(t)$, 
with $p^\kappa(y;t|\tau)$ the Gaussian densities corresponding to the SDE (\ref{eqOU0}).

In the derivation of (\ref{drift}), 
the first line follows from (\ref{def}) 
and the second from the application of the rule of the sum of Gaussian variables  
reminding that the process $\Ykappa$ is Gaussian, see (\ref{eqOU0}).
The third line contains the multiplication and division by $\sqrt{\sum_\kappa \E [(\Ykappa)^2]}$
and in the fourth $\tau_{\rm{eff}}$ is introduced
and the parameter $\sigma_0$ is moved below the square root.
The fifth line follows from (\ref{def}) and, finally, 
in the sixth line the rule of the sum of Gaussian variables is used again
in the first equality and (\ref{MOD1}) is used to have the last equality.

To conclude, by setting $\sum_\kappa dW_t^\kappa = dW_t^{\rm{eff}}$,
the stochastic process $\Vcm$ satisfies the following non-autonomous SDE
$$
d\Vcm= - \frac{\Vcm}{\taueff} dt + \frac{\sqrt{2 \sigma_0}}{M} \, dW_t^{\rm{eff}} \,. 
$$

\subsection{Proof of Theorem \ref{Th:ggBm1}}
By considering the scaled process $\Ykappa=\Ukappa/\sqrt{\Dkappa}$,
the stochastic process $\Vcm$ results to be
$$
\Vcm= \frac{\sqrt{\sigma_0}}{M} \sum_{\kappa} \Ykappa \,.
$$
Let us write
\be
\overline{Y}_t^N = \frac{1}{N} \sum_{\kappa=1}^N \Ykappa.
\ee
Recall that $Y^\kappa_t \stackrel{d}{=} Y^1_t$, $\forall\, \kappa$. For $x=0$, 
\be
\mathbb{E}_x[\overline{Y}^N_t] = x \, 
\mathbb{E}_0\left[ \frac{e^{-t/\tau}}{\sqrt{\sigma}} \right] = 0 \,,
\ee
and 
\be
Var_0[\overline{Y}^N_t] = \frac{1}{N} \,
\mathbb{E}_0 [\tau (1- e^{-2t/\tau})] = \frac{1}{N} \, \E_0[(Y^1_t)^2] \,.
\ee
From the Central Limit Theorem, as $N\to \infty$
$$
\frac{\overline{Y}^N_t - \mathbb{E}_0[\overline{Y}^N_t]}{\sqrt{Var_0[\overline{Y}^N_t]}} =  \sqrt{N} \frac{\overline{Y}_t^N}{\sqrt{\E_0[(Y^1_t)^2]}} = \frac{1}{\sqrt{\E_0[(Y^1_t)^2]}} \frac{1}{\sqrt{N}} \sum_{\kappa} \Ykappa \stackrel{d}{\to} W_1 \,,
$$
that is, as $N\to \infty$, 
$$
\frac{1}{\sqrt{N}} \sum_{\kappa=1}^N \Ykappa \stackrel{d}{\to} 
\sqrt{\E_0 [(Y^1_t)^2]} \, W_1 = B^H_t \,, \quad \forall \, t > 0 \,.
$$
Finally, by multiplying both side by $\sqrt{\sigma_0}/M$,
we obtain the following equivalence in distribution
\be 
\forall\, t>0 \,, \quad 
\frac{1}{\sqrt{N}} Z_t \stackrel{d}{\to} \frac{\sqrt{\sigma_0}}{M} B^H_t \,, 
\quad N \to \infty \,.
\ee
Moreover,  by setting $\sqrt{\sigma_0}/M=\sqrt{\Lambda}$,
the stochastic process $N^{-\frac{1}{2}} \Vcm$ 
is statistically equivalent in the sense of
the Central Limit Theorem to the process $\sqrt{\Lambda} \, B_t^H$, i.e.,
$$
\frac{1}{\sqrt{N}} \Vcm \to \sqrt{\Lambda} \, B_t^H \,.
$$

\subsection{Proof of Theorem \ref{Th:ggBm2}}
The proof of Theorem \ref{Th:ggBm2} 
follows by considering the previous theorem. 
Let us consider the Fourier transforms of both densities. 
We have that
\begin{align*}
\mathbb{E}_0[e^{i \xi Z_t}] = & \int e^{i \xi z} \frac{d}{dz} \mathbb{P}_0(Z_t < z) \, dz\\
= & \int e^{i \xi z} \frac{d}{dz} \mathbb{P}_0(\sqrt{\Lambda} B^H_t < \frac{z}{\sqrt{N}}) \, dz\\
= & \int e^{i \xi \sqrt{N} z} \frac{d}{dz} \mathbb{P}_0(\sqrt{\Lambda} B^H_t < z) \, dz\\
= & \mathbb{E}_0[e^{i \xi \sqrt{N\Lambda} B^H_t}] \,.
\end{align*}

\section{Numerical simulations}
\label{simulations}
In this Section we show by numerical simulations the equivalence in distribution of 
the three processes $\Vcm$ (\ref{MOD1}), $\Vcm^*$ (\ref{langevintype}) 
and $\Vcm^H$ (\ref{ggBmMURA}) as proved 
in Theorems \ref{Th:nonautonomous} and \ref{Th:ggBm2}.
The results of the simulations are shown in Fig.~\ref{figure1} 
by plotting the density functions of the 
three processes at different times.
Moreover, the insets show the corresponding variances. 

The centre-of-mass like process $\Vcm$ is simulated
by using its definition~\eqref{MOD1} 
through the process $X_t^\kappa$, which obeys Eq.~\eqref{OU}. 
We remind, that in~\eqref{OU}
the parameters $\tau^\kappa$ and $\sigma^\kappa$, see Eq.~\eqref{def}, 
follow the corresponding distributions $q(\tau)$ and $g(\sigma)$.
In particular, in \textit{Panel a)}, these distributions are 
\begin{equation}
\label{eq:par_a}
g(\sigma) = 
-\frac{3}{8} \frac{\sigma^{-3/4}}{\Gamma(-2/3)} \exp\left[-\sigma^{-3/8}\right] \,, 
\qquad q(\tau)=\exp(-\tau) \,,
\end{equation}
where we remind that $\Gamma(-2/3) < 0$,
and in \textit{Panel b)} 
\begin{equation}
\label{eq:par_b}
g(\sigma)=\eta \frac{\sigma^{\nu - 1}}{\Gamma(\nu/\eta)} \exp(-\sigma^\eta) \,, 
\qquad q(\tau)=\frac{\alpha}{\Gamma(1/\alpha)}\frac{1}{\tau}L^{-\alpha}_\alpha\left(\tau\right) \,,
\end{equation}
where $L^{-\alpha}_\alpha\left(\tau\right)$ is the L\'evy extremal density and $\nu=0.5$, $\eta=1.3$, $\gamma_0=10^{-5}$, $\alpha=0.75$. 

The process $\Vcm^H$ is simulated by 
\be
Z_t^H=\sqrt{N\Lambda} \, B_t^H=\sqrt{N\Lambda}\sum_\kappa{Y_t^\kappa}=
\frac{\sqrt{N\sigma_0}}{M} \sum_\kappa{Y_t^\kappa} \,,
\ee 
where $Y_t^\kappa$ obeys Eq.~\eqref{YOU}.

Finally, the process $Z_t^*$ is simulated by the SDE
\[
dZ_t^*= - \frac{Z_t^*}{\taueff(t)} 
dt + \frac{\sqrt{2 N\sigma_0}}{M} \, dW_t \,,
\]
where, in analogy with previous cases,
we use the properties of the sum of Gaussian variables to set
$dW_t^{\mathrm{eff}}=\sum_\kappa dW_t^\kappa=\sqrt{N} dW_t$.
Since we already have the trajectories $Y_t^\kappa$ that 
we simulated for the process $\Vcm^H$,
for numerical convenience, the computation of $\tau_\mathrm{eff}$ 
is performed according to definition in terms of series given in Eq.~\eqref{gammaeff}. 
Moreover,
we set $\sigma_0 = 1/N$ because in the simulations we sum $N$ trajectories. 

Simulations support the equivalence in distribution among the three processes,
at least for the chosen distributions of the parameters.

\begin{figure}[!h]
\includegraphics[width=0.48\linewidth]{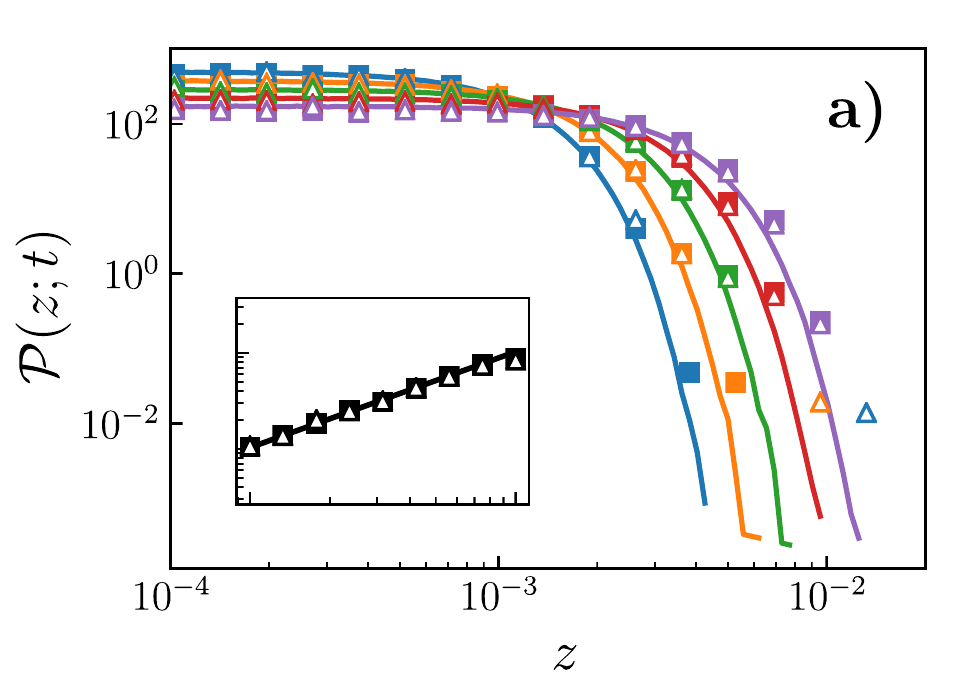}
\includegraphics[width=0.48\linewidth]{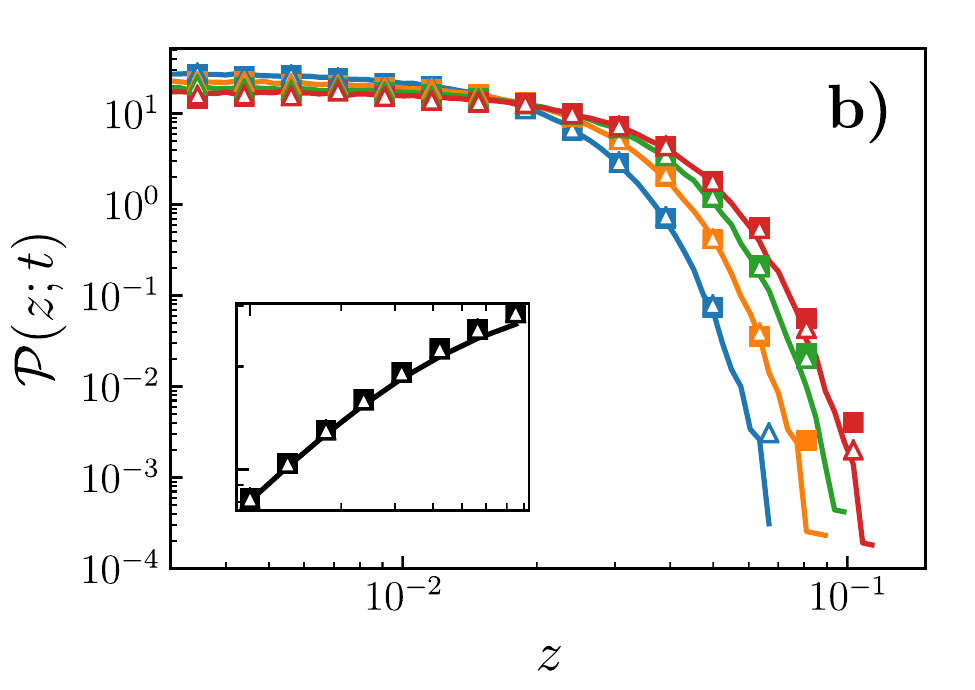}
\caption{(colour online) Equivalence in distribution of the processes $Z_t$ ($\blacksquare$), $Z_t^H$ ($\vartriangle$) and $Z_t^*$ (solid lines). Different colours represent different times.
Insets: Log-log plots of the
variances of the corresponding processes as functions of time. 
The distributions of the parameters are given with Eq.~\eqref{eq:par_a} for Panel a) and Eq.~\eqref{eq:par_b} for Panel b).
}
\label{figure1}
\end{figure} 

\section{Conclusions and perspectives for applications}
\label{conclusions}
In this paper we have considered an ensemble of OU processes characterized
by a population of relaxation times and of noise amplitudes. 
In the framework of the Brownian motion, the OU process can be viewed as a 
stochastic dynamical equation and is called Langevin equation.
In this sense, we have considered 
a heterogeneous ensemble of Brownian particles.

In particular, we have studied 
the stochastic dynamics of the centre-of-mass like process of this ensemble
and proved that this process
is equivalent in distribution to a process generated by a non-autonomous SDE
with time-dependent drift term and uncorrelated white noise, i.e.,
no-memory effects emerge in opposition 
to approaches based on the generalized Langevin equation  
or on Langevin equations with coloured noises,
and it is equivalent in distribution also to a randomly-scaled Gaussian process,
i.e., the product of a Gaussian process times a non-negative random variable.

Moreover we showed that
{\it anomalous diffusion} emerges as a consequence of the heterogeneity.
In fact, we found that
the population of the relaxation times contributes 
to the emergence of an anomalous scaling in time, as shown in (\ref{correlazione}),
and the population of the noise amplitudes contributes 
to have a non-Gaussian density function (\ref{ggBmpdfn}). 

The present study is a further step towards the interpretation
of anomalous diffusion as a consequence of a complex environment
\cite{pagnini-pa-2014,molina_etal-pre-2016,pagnini_etal-fcaa-2016} 
that is here represented by the heterogeneity of the ensemble of particles.
The improvement provided with this paper lays in the fact  
by using the OU process we establish a relation 
between the anomalous diffusion and the Brownian motion modelled by the Langevin equation,
which is a dynamical equation, so put the emergence of anomalous diffusion
on physical basis.
Then, future developments will concern the study of the motion
and the diffusion properties
of a test-particle immersed in a complex environment
embodied by such heterogeneous ensemble. 
This has a direct application for modelling anomalous diffusion 
\cite{hofling_etal-rpp-2013}
and collective motion \cite{mayor_etal-nr-2016} in biological systems.
In fact, anomalous diffusion in biological systems is largely observed in experiments
and it is due to macromolecular crowding in the interior of cells 
and in cellular membranes, because of their densely packed and heterogeneous structures
\cite{hofling_etal-rpp-2013,regner_etal-bj-2013}. 
The proposed modelling approach can be applied for modelling 
diffusion of meso-scopic test-particle in an environment also composed 
of meso-scopic particles embodying the heterogeneous and 
crowded environment where the test-particle moves.
The study of diffusion of a test-particle in a heterogeneous and crowed
environment motivates future developments of the present research.

\section*{Acknowledgments}
This research is supported by the Basque Government through the 
BERC 2014-2017 program and the BERC 2018-2021 program, 
and by the Spanish Ministry of Economy and Competitiveness MINECO 
through BCAM Severo Ochoa excellence accreditation SEV-2013-0323 
and through project MTM2016-76016-R "MIP".
VS acknowledges BCAM, Bilbao, for the financial support to her internship research period,
and SV acknowledges the University of Bologna for the financial support through
the "Marco Polo Programme" for funding her PhD research period abroad spent at BCAM, Bilbao.



\section*{References}






\end{document}